\documentclass[journal,10pt]{IEEEtran}
\usepackage{amsfonts}
\usepackage{amssymb}
\usepackage{amsmath}
\usepackage{amsthm}
\usepackage{graphicx}
\usepackage{subfigure}
\usepackage{cite}
\usepackage{setspace}
\usepackage{algorithm}
\usepackage{algorithmic}
\usepackage{multirow}
\usepackage{amsmath}
\usepackage{xcolor}
\usepackage{subfigure}
\usepackage{subfig}
\usepackage{stfloats}
\usepackage{mathrsfs}
\usepackage[abs]{overpic}

\newtheorem{lemma}{Lemma}
\newtheorem{corollary}{Corollary}
\begin{document}

\title{Tradeoff between Ergodic Rate and Delivery Latency in Fog Radio Access Networks }

\author{Bonan Yin, Mugen Peng,~\IEEEmembership{Fellow,~IEEE}, Shi Yan,~\IEEEmembership{Member,~IEEE} and Chunjing Hu\\
\thanks{This work was supported in part by the National Natural Science Foundation of China under No. 61671074 and 61831002, State Major Science and Technology Special Project under 2018ZX03001023, and 2018ZX03001025.  \emph{(Corresponding author: Mugen Peng, Shi Yan.)}

The authors are with the State Key Laboratory of Networking and Switching Technology, Beijing University of Posts and Telecommunications, Beijing 100876, China(e-mail: yinbonan@bupt.edu.cn; pmg@bupt.edu.cn; yanshi01@bupt.edu.cn; cjhu@bupt.edu.cn).}}

\maketitle
\begin{abstract}
Wireless content caching has recently been considered as an efficient way in fog radio access networks (F-RANs) to alleviate the heavy burden on capacity-limited fronthaul links and reduce delivery latency. In this paper, an advanced minimal delay association policy is proposed to minimize latency while guaranteeing spectral efficiency in F-RANs. By utilizing stochastic geometry and queueing theory, closed-form expressions of successful delivery probability, average ergodic rate, and average delivery latency are derived, where both the traditional association policy based on accessing the base station with maximal received power and the proposed minimal delay association policy are concerned. Impacts of key operating parameters on the aforementioned performance metrics are exploited. It is shown that the proposed association policy has a better delivery latency than the traditional association policy. Increasing the cache size of fog-computing based access points (F-APs) can more significantly reduce average delivery latency, compared with increasing the density of F-APs. Meanwhile, the latter comes at the expense of decreasing average ergodic rate. This implies the deployment of large cache size at F-APs rather than high density of F-APs can promote performance effectively in F-RANs.
\end{abstract}

\section{Introduction}

Fundamental shift from voices and messages to rich media applications makes data traffic increase rapidly and quality of experience (QoE) restrictive, which puts extreme pressure on the current radio access networks (RANs). The total throughput of mobile networks in 2020 is expected to become 1000-fold larger than that in 2010 [1]. Meanwhile, key performance indicators are realized in the fifth generation (5G) cellular network, including 1000-fold area capacity, 100-fold energy efficiency (EE), 10-20 Gbps peak rate, milliseconds end-to-end latency, trillions of devices connections, as well as ultra reliability [2].

With the vision of meeting demands of huge capacity and massive connections, a huge number of small cell base stations (BSs) need to be additionally deployed in existing cellular network, which results in increasingly high capital expenditure (CAPEX) and operation expenditure (OPEX), as well as high energy consumption. Meanwhile, spectral efficiency (SE) is reduced severely due to the increasing inter-cell interference. Motivated by the reduction of both CAPEX and OPEX as well as the improvement of SE and EE, cloud radio access network (C-RAN) is presented as a potential solution and thus has attracted widespread attention in both industry and academic community. Through splitting the functions of BSs into centralized baseband unit (BBU) pool and remote radio heads (RRHs), collaboration radio signal processing (CRSP) and cooperative radio resource management (CRRM) can be realized, which improve SE in an energy-efficiency and scalable way. Meanwhile, utilizing the dense deployment of RRHs instead of small cell BSs, trillions of devices are envisioned to be provided with fantastic service quality, in which massive connections can be efficiently supported. However, centralized C-RANs put all functions of the air interface at BBU pool, resulting in long latency and low reliability with capacity-limited fronthaul links.

Edge caching is proposed as a promising technique to support ultra-low latency and high reliability requirements in 5G, especially for applications such as augmented reality/virtual reality (AR/VR) and vehicular communications. In particular, with the concept of decentralization, by prefetching popular contents during off-peak time at the edge of wireless networks, caching can alleviate peak-hour network congestion, mitigate heavy burden on the capacity-limited fronthaul links, and improve users' QoEs. Furthermore, deploying caches closer to the edge of networks, more gains of latency and reliability can be achieved. Nevertheless, edge nodes design transmitted radio signals based only on the local caches in a distributed way, which hinders the advantage of collaborative processing in BBU pool.

Motivated by the integration of benefits for both centralized C-RAN architecture and edge caching, a cache-enabled fog radio access network (F-RAN) has been presented to reduce latency and improve reliability while ensuring sufficiently high SE [3].
In F-RANs, through equipped with storage and computing capacities, traditional access points (APs) and user equipments (UEs) are evolved to the fog-computing based access points (F-APs) and fog-computing based user equipments (F-UEs) respectively, which contributes to execute local cooperative signal processing at the edge of networks. Meanwhile, CRSP and CRRM functionalities can also be performed in centralized BBU pool, which ensures acceptable SE in a cost-efficiency way. Due to the aforementioned embedded qualities, users' services can both be executed in a centralized form at BBU pool through fronthaul links, and be executed in a distributed form at local F-APs or F-UEs at the edge of networks. Since a part of users no longer directly access to BBU pool through fronthaul links, the heavy burdens on both capacity-limited fronthaul links and BBU pool are alleviated, and delivery latency can be reduced significantly as well.

\subsection{Related Work}

The system architecture and key techniques of F-RANs have been proposed in [3]. Particularly, different caching strategies for F-RANs, such as cache most popular, cache distinct, and fractional cache distinct, has been discussed in [4], which are suitable for centralized request, discrete request, and capacity-limited fronthaul scenario, respectively. Meanwhile, significant gains of SE [5], EE [6] and latency [7] have been obtained in cache-enabled F-RANs. The major reason is that requested service contents are delivered to users with limited edge caches, without passing through capacity-limited fronthaul or backhaul links, contributing to a little of energy consumption and a low delivery latency. On the other hand, nearer device-to-device (D2D) users or F-APs bring higher signal to interference plus noise ratio (SINR) to users and thus providing high SE and EE. Unfortunately, the theoretical performance analysis of outage probability, ergodic rate and delivery latency in cache-enabled F-RANs is still not straightforward, which constrains the further investigation of F-RANs with cache.

Actually, there have been series of publications focusing on the theoretical analysis of outage probability and ergodic rate in traditional cellular networks. Utilizing the tool of stochastic geometry, the closed-form expressions of outage probability and average delivery rate in small base stations (SBSs) have been derived in [8], where SBSs are distributed according to homogeneous poisson point process (PPP). In [9], the closed-form expression of outage probability by jointly considering spectrum allocation and storage constraints has been derived in a two-tier PPP-based HetNets. A general expression of successful delivery probability in \emph{N}-tier HetNets has been derived in [10], where the density and cache size of BSs have significant impacts on successful delivery probability. The authors have presented analytical results on successful delivery probability based on different transmission schemes in a cluster-centric small cell network, and contributed to analyze the tradeoff between transmission diversity and content diversity in [11]. In addition, the authors have investigated the coverage probability and ergodic rate in a PPP-based F-RANs with D2D users in [12][13].

It is worth noting that one of the most challenging requirements in 5G is ultra-low end-to-end latency. However, with the assist of caching technique, it is still hard to find the accurate theoretical analysis of transmit latency, as well as the tradeoff between transmit latency and other traditional performance metrics with the impacts of key operating parameters in the existing publications. Meanwhile, well-known comprehensive evaluation metrics need to be proposed, which may highlight the characterization of latency.

Recently, there have been some comprehensive evaluations with regard to latency performance. In [13], tractable expressions of the effective capacity, which reflects both latency and capacity performance, are derived in PPP-based C-RANs. Normalized delivery time is presented from the aspect of fundamental information-theoretic, which reflects the interplay between cloud processing and edge caching in F-RANs [14]. Nevertheless, these aforementioned metrics have not been widely recognized, and need to be further researched and clarified in cache-enabled F-RANs. At the same time, more comprehensive evaluation metric is expected to analyze the network performance under 5G and beyond systems. On the other hand, by using stochastic geometry and queueing theory, average ergodic rate, outage probability and transmit latency are derived in a 3-tier PPP-based HetNet in [15]. In cluster-centric HPPP-based F-RANs, explicit expressions of ergodic rate and transmit latency are provided with a hierarchical content caching scheme [16]. However, research on transmit latency performance with the assist of caching is still in its infancy, especially for F-RANs. There is an urgent need to derive and analyze the closed-form or asymptotic solutions of the theoretical performance under different evolutionary F-RANs to reveal the intrinsic correlation of ergodic rate, transmit latency, connections, and etc.

\subsection{Contributions}

Since the tradeoff between ergodic rate and delivery latency in cache-enabled F-RANs is still not straightforward, characteristic of F-RANs and the corresponding performance analysis have been concerned in this paper. Meanwhile, key factors impacting on alleviating capacity constraints of fronthaul links and reducing latency have been exploited. The main contributions are summarized as follows:

\begin{itemize}
\item In order to highlight the metric of delivery latency and take advantage of edge caching, an advanced minimal delay association policy is presented to minimize latency while guaranteeing SE in F-RANs, where users are associated with the desired BS according to delivery latency. Particularly, the delivery latency can be converted to requested SINR at the corresponding BS on the basis of caching state. Furthermore, owing to the fronthaul and backhaul link latency, users need quite large SINR to associate with RRHs and F-APs without requested contents, resulting in that users prefer to access to F-APs with requested contents. As a result, users are associated according to SINR as well as caching state at the corresponding BS.
\item Theoretical performance gain of the proposed minimal delay association policy in F-RANs is clarified. By using stochastic geometry and queueing theory, closed-form expressions of successful delivery probability, average ergodic rate, delivery latency under the concerned association policies are derived. Furthermore, for intuitive presentation, the traditional association policy to maximize reference signal receiving power (RSRP) is utilized as a contrast and baseline.
\item To evaluate internal correlation between delivery latency and other performance metrics in F-RANs, the impacts of key operating parameters on the aforementioned performance metrics are investigated. Meanwhile, the tradeoff between average ergodic rate and delivery latency is exploited, and numerical results demonstrate that increasing the cache size of F-APs can more significantly reduce average delivery latency, compared with increasing the density of F-APs, which comes at the expense of decreasing average ergodic rate.
\end{itemize}

The remainder of this paper is organized as follows. In Section II, the system model with cache protocol and user association policy is presented. In Section III, closed-form expressions in terms of successful delivery probability, ergodic rate, and delivery latency based on the traditional maximal RSRP association policy are derived. While an advance minimal delay association policy and related theoretical analysis are presented in Section IV. Simulation results are provided in Section V. Finally, conclusions are drawn in Section VI.

\section{System Model}
\subsection{Network Model}

As illustrated in Fig. 1, a group of RRHs and F-APs are spatially deployed in a two-dimensional disc plane $D^2$ according to an independent homogeneous PPP, denoted as $\Phi_i$ with density of $\lambda_i$ for $i \in \left\{ {R, F} \right\}$ respectively, which characterizes the randomness of network topology [23]. Meanwhile, users are also modeled as a PPP distribution, denoted as $\Phi_u$ with constant density of $\lambda_u$. Note that RRHs are connected to virtual BBU pool through fronthaul links, meanwhile F-APs are connected to the core network through network controller with the assist of the backhaul links. Edge cache is deployed in each F-AP, where users can achieve requested contents directly from limited caches in F-APs if the requested contents are cached locally. This content delivery approach operates in two phases, namely pre-fetching and delivery phase, where pre-fetching phase is executed at a large time scale corresponding to the period of fixed file popularity and always during the off-peak time. Otherwise, if there are not the requested contents in local limited caches, users can acquire services at RRHs or F-APs, which utilizes fronthaul links to fetch the requested contents from centralized BBU pool or utilizes backhaul links to fetch the requested contents from core network, respectively.

For the wireless channel, both large-scale fading and small-scale fading are considered. Specifically, the large-scale fading is modeled by a standard distance-dependent path loss attenuation $r^{-\alpha}$ with path loss exponent $\alpha$, where $r$ denotes the distance between certain BS and the corresponding user [17]. While the small-scale fading is assumed to be denoted as the Rayleigh fading, i.e., $h \sim CN\left( {0,1} \right)$. Each user experiences an additive noise that obeys zero-mean complex Gaussian distribution with variance $\sigma^2$ [18].
\subsection{Cache Model}
For the content delivery application, a database consisting of $N$ contents is considered to be deployed in BBU pool, and all contents are assumed to have equal length $L$. Each F-AP has limited caching capacity with storage size $C_F$, where $C_F<N \times L$. Therefore, only part of database contents can be cached in each F-AP. For simplicity, each F-AP is assumed to just cache integer contents and thus $C_F$ is integral multiple of content length $L$.
Let $\mathcal{C}=\{f_1, f_2, \cdots, f_M\}$ denote the set of cached contents in F-APs, where $M=C_F/L$ means the number of cached contents.

\begin{figure}[t]
  \centering
  \includegraphics[width=0.5\textwidth]{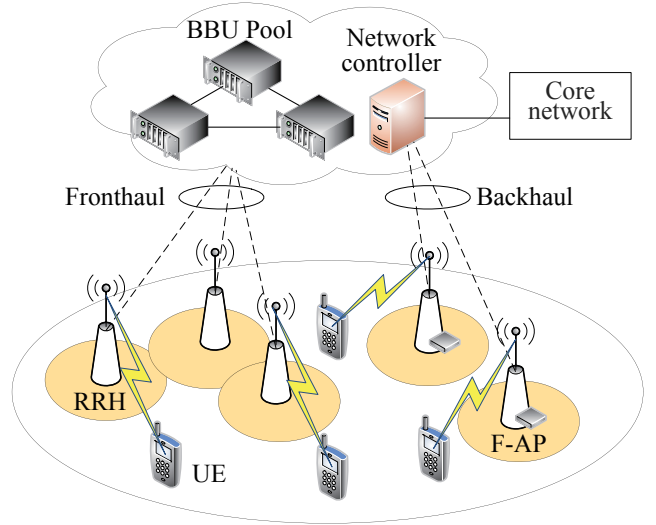}
  \vspace*{-5pt} \caption{System model of cache-enabled F-RANs }\vspace*{-10pt}
  \label{1}
\end{figure}

Actually, users get strong desire of the most popular contents, resulting in that only a small portion of the $N$ contents are frequently requested by the majority of users. It can be assumed that content popularity follows Zipf distribution [13], and demand probability that the $i-th$ most popularity content is requested can be expressed as
\begin{equation}\label{file_probability}
f_i\left( \tau,N \right)=\frac{1/{i^{\tau}}}{\sum\nolimits_{j = 1}^N {1/{j^\tau }}},
\end{equation}
where the content with a smaller index has a larger probability to be requested by users, i.e., $f_i\left( \tau,N \right)>f_j\left( \tau,N \right)$ if $i<j$. Note that Zipf exponent $\tau$ models the popularity skewness of contents, which is a nonnegative number. In addition, only a fewer of popular contents are frequently requested by users when Zipf exponent $\tau$ become larger. Specially, content popularity follows the uniform distribution as $\tau=0$.

Caching hit probability is defined as the probability that typical user finds requested content $f$ in the corresponding cache, i.e.,
$p_{hit}=\Pr\left(f\in \mathcal{C}\right)$.
Particularly, most popular caching strategy is adopted here, where F-APs only cache the most popular contents. It's reasonable for services with high Zipf exponent $\tau$, such as HD video service, since only a few of popular contents are frequently requested by major users. Therefore, the caching hit probability is denoted as
\begin{equation}\label{file_probability}
p_{hit}=\sum\nolimits_{i = 1}^{M}f_i\left( \tau,N \right).
\end{equation}
Note that other caching strategies, suitable for their specific services, are also applicable in this paper, where only $p_{hit}$ and densities of BSs with related contents need to be modified.
\subsection{Access Model}
Without any loss of generality, typical user under consideration is assumed to be located at the origin. The user association policy does not only depend on the signal receiving power but also the requested and cached contents. Here, two different user association policies are considered. Firstly, typical user is associated with the strongest BS based on RSRP, thus the user association policy is given by:
\begin{itemize}
\item Typical user is associated with F-APs if $k{r_{x_0}}^ {-\alpha} \le {r_{y_0}}^ {-\alpha}$
\item Typical user is associated with RRHs if $k{r_{x_0}}^ {-\alpha} > {r_{y_0}}^ {-\alpha}$
\end{itemize}
where $r_{x_0}$ is the distance between user and the closest F-AP, while $r_{y_0}$ is the distance between user and the closest RRH. Note that $k < 1$ is assumed since $P_R < P_F$ in general,
where $P_F$ and $P_R$ denote the transmit powers at F-APs and RRHs, respectively.

Secondly, typical user is connected to BS with the lowest delivery latency, thus the user association policy is given by:
\begin{equation}\label{file_probability}
l\left( j \right)=\arg\mathop {\max }\limits_{i \in \left\{ {{F_c},{{\tilde F}_c},R} \right\}}
\left\{ {D_{i|j}} \right\},
\end{equation}
where $F_c$, ${\tilde F}_c$ and $R$ denote set of F-APs with the requested content, set of F-APs without the requested content and set of RRHs.
In detail, the latency for typical user at F-APs with the requested content $j$, i.e., $D_{F_c|j}$, consists of queuing latency at F-APs and transmit latency of F-APs-user links, while the latency for typical user at F-APs without the requested content $j$, i.e., $D_{{\tilde F}_c|j}$, also need to consider the latency of backhaul links. In addition, the latency for RRHs user $D_{R|j}$ includes the latency of fronthaul links , the queuing latency at RRHs and the transmit latency of RRHs-user links.

\subsection{Signal to Interference plus Noise Ratio}
Given that typical user is associated with F-APs, then received instantaneous SINR at the typical user is given by
\begin{equation}\label{file_probability}
\gamma_{x_0}^{F}=\frac{P_F{\left| h_{x_0} \right|}^2{r_{x_0}}^ {-\alpha}} {\sum\limits_{x \in {\Phi _{F} \backslash x_0}}{P_F}{\left| h_{x} \right|}^2{r_x}^ {-\alpha}+\sum\limits_{y \in {\Phi _R}}{P_R}{\left| h_{y} \right|}^2{r_y}^ {-\alpha}+{\sigma}^2},
\end{equation}
where ${\left| h \right|}^2 $ characterizes the Rayleigh small-scale fading channel gain between typical user and the serving or interfering BS, i.e., ${\left| h \right|}^2 \sim exp\left( 1 \right)$. $r^{-\alpha}$ denotes the path loss, where $\alpha>2$ is the path loss exponent. $I_{F}=\sum\nolimits_{x \in {\Phi _{F} \backslash x_0}}{P_F}{\left| h_{x} \right|}^2{r_x}^ {-\alpha}$ denotes intra-tier interference from other F-APs except the serving F-AP. $I_R=\sum\nolimits_{y \in {\Phi _R}}{P_R}{\left| h_{y} \right|}^2{r_y}^ {-\alpha}$ denotes inter-tier interference from RRHs.

Otherwise, typical user is served by RRHs, limited interference collaboration is considered here since the available channel state information (CSI) are always inaccuracy in realistic scenarios. With limited feedback, the channel direction information (CDI) is fed back through a quantization codebook known at both the transmitter and receiver. The quantization is chosen from a codebook of unit norm vectors of size $L=2^B$, where $B$ is the number of feedback bits [19].

With the limited feedback, the statistic of the quantized CDI is considered. Let $cos{\theta _{x}} = \left| {\tilde h_{x}^*{{\hat h}_{x}}} \right|$, where ${\theta _{x}} = \angle \left( {\tilde h_{x}{{\hat h}_{x}}} \right)$
denotes the angle between the normalized version and the quantization of channel vector $h_x$, and $h_x^*$ refers to the conjugate transpose, or Hermitian, of channel vector $h_x$.
Thus the quantized coefficient can be expressed as
\begin{equation}\label{file_probability}
{\zeta} = 1 - {L}\beta \left( {{L},\frac{{{N_B}}}{{{N_B} - 1}}} \right),
\end{equation}
where $\beta\left(x, y\right)$ is the Beta function, i.e., $\beta\left(x, y\right) = {\textstyle{ \Gamma(x)\Gamma(y)  \over {\gamma \left( x+y \right)}}}$ with $\Gamma(x) =\int_0^\infty t^{x-1}e^{-t}dt$ as the Gamma function, and $N_B$ is the number of antennas at each RRH.

Meanwhile, with regard to the interference signals from other RRHs except the serving RRH, the interference power coefficient with the limited feedback is denoted as
\begin{equation}\label{file_probability}
\upsilon  = {2^{ - \frac{{{B}}}{{{N_B} - 1}}}}.
\end{equation}

Therefore, the received instantaneous SINR under limited interference collaboration at typical user is statistically equivalent to
\begin{equation}\label{file_probability}
\gamma_{y_0}^{R}=\frac{ P_R \zeta {\left| h_{y_0} \right|}^2{r_{y_0}}^ {-\alpha}} {\sum\limits_{x \in {\Phi _{F}}}{P_F}{\left| h_{x} \right|}^2{r_x}^ {-\alpha}+\upsilon \sum\limits_{y \in {\Phi _{I_{R}\backslash{y_0}}}}{P_R}{\left| h_{x} \right|}^2{r_y}^ {-\alpha}+{\sigma}^2},
\end{equation}
where $I_{F}=\sum\nolimits_{x \in {\Phi _{F}}}{P_F}{\left| h_{x} \right|}^2{r_x}^ {-\alpha}$ denotes inter-tier interference from F-APs, $I_{R}=\upsilon\sum\nolimits_{x \in {\Phi _{I_{R}\backslash{y_0}}}}{P_R}{\left| h_{x} \right|}^2{r_y}^ {-\alpha}$ denotes redundant interference from RRHs except the serving RRH.
\section{Maximal RSRP Association Policy}
In this section, we derive successful delivery probability, average ergodic rate, and delivery latency for cache-enabled F-RANs based on the traditional maximal RSRP association policy, where typical user is associated with the strongest BS according to RSRP. The probability of tier association of typical user is investigated firstly.

\subsection{Probability of Tier Association}
\begin{lemma}
The user association probability based on maximal RSRP association policy at F-APs is given by [20]
\begin{equation}\label{file_probability}
A_F=\frac{\lambda_F}{\lambda_F+ k^2\lambda_R},
\end{equation}
and the user association probability based on maximal RSRP association policy at RRHs is given by
\begin{equation}\label{file_probability}
 A_R=1-A_F=\frac{k^2\lambda_R}{\lambda_F+ k^2\lambda_R},
\end{equation}
where $k$ is set by $k=\left( P_R/P_F \right)^{1/\alpha}$ to maximize the success delivery probability of F-APs user to take full advantage of caching and minimize delivery latency, and this is at the expense of a lower success delivery probability of RRHs user, compared with the maximum value.
\end{lemma}

Note that the user association probabilities based on maximal RSRP association policy are independent of the requested content and caching allocation, while only depend on the density and power of both F-APs and RRHs.

\subsection{Successful Delivery Probability}
Success delivery probability is defined as the probability of SINR between the associated BS and typical user larger than a SINR threshold $\delta$, which is also called coverage probability. Considering that typical user can be associated with certain BS under open access mode, the success delivery probability of typical user based on maximal RSRP association policy is expressed as
\begin{equation}\label{file_probability}
\begin{aligned}
 S\left( \delta,\alpha \right) &=Pr \left( SINR>\delta \right)\\
 &=A_F S_F\left( \delta,\alpha \right) + A_R S_R\left( \delta,\alpha \right),
\end{aligned}
\end{equation}
where $A_F$ and $A_R$ denote user association probability at F-APs and RRHs, respectively, $S_F\left( \delta,\alpha \right)$ and $S_R\left( \delta,\alpha \right)$ are success delivery probability of typical user at F-APs and RRHs, respectively.

Note that the success delivery probability of typical user based on maximal RSRP association policy is independent of the requested content and caching allocation, because of the user association policy.

\begin{lemma}
The success delivery probability of typical user at F-APs based on maximal RSRP association policy is derived as
\begin{equation}\label{file_probability}
S_F\left( \delta,\alpha \right) =\frac{\lambda_F+k^2 \lambda_R} {\lambda_F \left( 1+ \rho \left( \delta,\alpha \right) \right)+ k^2 \lambda_R \left( 1+ \rho \left( \frac{\delta P_R}{k^{\alpha}P_F},\alpha \right) \right) },
\end{equation}
where $\rho \left( \delta,\alpha \right)=\int_{{\delta ^{ - 2 / \alpha}}}^\infty  {\frac{{{\delta ^{2/\alpha }}}}{{1 + {u^{\alpha /2}}}}} du$
and the success delivery probability of typical user at RRHs based on maximal RSRP association policy is derived as
\begin{equation}\label{file_probability}
S_R\left( \delta,\alpha \right) =\frac{\lambda_F+k^2 \lambda_R} {\lambda_F \left( 1+ \rho \left( \delta_1,\alpha \right) \right)+ k^2 \lambda_R \left( 1+ \rho \left( \delta_2,\alpha \right) \right) },
\end{equation}
where $\delta_1=\frac{\delta k^{\alpha}P_F}{\zeta P_R}$, $\delta_2=\frac{\upsilon \delta } {\zeta}$, $\zeta$ and $\upsilon$ are the statistic of the quantized coefficient and the interference power coefficient, respectively.
\end{lemma}
\begin{proof}
See Appendix A.
\end{proof}

Note that the success delivery probabilities based on maximal RSRP association policy are related to the density and power of both the F-APs and RRHs, as well as the SINR threshold, except for the size and location of caches. Particularly, the success delivery probability of typical user at F-APs can be simplified as $1/(1+\rho \left( \delta,\alpha \right))$ and has nothing to do with density and power of BSs, since $k$ is set by $k=\left( P_R/P_F \right)^{1/\alpha}$, i.e., F-RAN degenerates into a single-tier network from the perspective of F-AP users.

This result can be simplified further and reduces to a simple closed-form expression with path loss exponent $\alpha=4$, which is given by Corollary 1.

\begin{corollary}
In the particular case of path loss exponent $\alpha=4$, the success delivery probabilities can be simplified as (13) and (14) shown in the bottom at next page, respectively.
\end{corollary}

\setcounter{equation}{12}
\begin{figure*}[hb]
\hrulefill
\begin{equation}\label{file_probability}
S_F\left( \delta,\alpha \right) =\frac{\lambda_F+k^2 \lambda_R} {\lambda_F \left( 1+ \sqrt {\delta}\left(\frac{\pi}{2}-arctan\left(\frac{1}{\sqrt {\delta}}\right)\right) \right)+ k^2 \lambda_R \left( 1+\sqrt {\frac{\delta P_R}{k^{\alpha}P_F}}\left(\frac{\pi}{2}-arctan\left({\sqrt {\frac{k^{\alpha}P_F}{\delta P_R}}}\right)\right) \right) },
\end{equation}
\end{figure*}
\setcounter{equation}{13}
\begin{figure*}[hb]
\begin{equation}\label{file_probability}
S_R\left( \delta,\alpha \right) =\frac{\lambda_F+k^2 \lambda_R} {\lambda_F \left( 1+\sqrt {\frac{\delta k^{\alpha}P_F}{\zeta P_R}}\left(\frac{\pi}{2}-arctan\left(\sqrt {\frac{\zeta P_R}{\delta k^{\alpha}P_F}}\right)\right) \right)+ k^2 \lambda_R \left( 1+ \sqrt {\frac{\upsilon \delta } {\zeta}}\left(\frac{\pi}{2}-arctan\left(\sqrt {\frac{\zeta}{\upsilon \delta } }\right)\right) \right) }.
\end{equation}
\end{figure*}

\subsection{Average Ergodic Rate}
Assuming that adaptive modulation and coding are utilized, and each user can achieve their rate with instantaneous SINR according to the Shannon bound, i.e. $\log(1 + SINR)$.

\begin{lemma}
The average ergodic rate of typical user with the associated BS based on maximal RSRP association policy is defined as [21]
\begin{equation}\label{file_probability}
\mathbb{E}[R_x]=\int_0^\infty {S_x\left( {2^t-1},\alpha \right)}dt,
\end{equation}
where $S_x\left( {2^t-1},\alpha \right) $ is the success delivery probability of typical user at certain BS $x$ by fixing $\delta = 2^t-1$.
\end{lemma}

\begin{proof}
The instantaneous rate of typical user with associated BS can be denoted as a random variable $R_x=\log(1 + \gamma_x)$ in bps/Hz with respect to the instantaneous SINR at BS $x$. Since $R_x$ is a positive random variable, the ergodic rate of typical user with associated BS based on maximal RSRP association policy is derived as
\begin{equation}\label{file_probability}
\begin{aligned}
&\mathbb{E}\left[R_x\right]=\int_0^\infty P\left(R_x \ge t\right)dt\\
&=\int_0^\infty P\left( \gamma_x \ge 2^t-1\right)dt=\int_0^\infty S_x\left( 2^t-1,\alpha \right)dt,
\end{aligned}
\end{equation}
and the proof is complete.
\end{proof}

Furthermore, according to both the success delivery probability of F-APs user and RRHs user, the average ergodic rate of F-RANs based on maximal RSRP association policy can be obtained by Corollary 2, which is the weighted average of ergodic rate of F-APs user and RRHs user.

\begin{corollary}
The average ergodic rate of typical user based on maximal RSRP association policy can be expressed as
\begin{equation}\label{file_probability}
\begin{aligned}
&R\left( \alpha \right)= A_F \mathbb{E} \left[ log_2 \left( 1+ \gamma_{x_0}^F\right) \right] +  A_R \mathbb{E} \left[ log_2 \left( 1+\gamma_{y_0}^R\right) \right]\\
&= A_F \int_0^\infty {S_F\left( {2^t-1},\alpha \right)}dt +A_R \int_0^\infty {S_R\left( {2^t-1},\alpha \right)}dt,
\end{aligned}
\end{equation}
where $S_F\left( {2^t-1},\alpha \right) $ and $S_R\left( {2^t-1},\alpha \right) $ are the success delivery probability of typical user at F-APs and RRHs by fixing $\delta = 2^t-1$, respectively.
\end{corollary}

\subsection{Delivery Latency}
The mean delivery latency of typical user at certain BS consists of two parts. The first part is transmission time $T$, which is the duration from the moment when transmission starts to the moment when the transmission ends. The second part is waiting time $W$, which is the duration from the moment of requested content arrivals to the moment when the transmission starts. Here, each BS is considered to be modeled as a $M/D/1$ queuing system. It characterizes transmission time interval with exponential distribution and fixed size data, which is suitable for HD video services in this paper [22].

\begin{lemma}
Since the $M/D/1$ queuing system is considered here, the mean delivery latency of typical user at certain BS $x$ is given by
\begin{equation}\label{file_probability}
\begin{aligned}
\bar D_x&= \mathbb{E} \left[ \frac{L}{R_x}+\frac{L}{R_x}\frac{\rho_x}{\left(1-\rho_x \right)} \right]\\
&\mathop  \ge \limits^{\left( a \right)} \underbrace {\frac{L}{\mathbb{E} \left[R_x\right]}}_T+\underbrace {\frac{L}{\mathbb{E} \left[R_x\right]}\mathbb{E} \left[\frac{\rho_x}{\left(1-\rho_x \right)}\right]}_W\\
&=\frac{L}{\mathbb{E}[R_x]-\rho'_x},
\end{aligned}
\end{equation}
which is the lower bound of the mean delivery latency of typical user at certain BS. The first part of equation (a) is the transmission time $T$ and the second part is the waiting time $W$. $\rho_x= N_x \xi L/R_x$ is denoted as the traffic intensity of BS $x$, where $N_x$ is the number of users associated with BS $x$, and $\xi$ is the service requested rate of each user. Similarly, $\rho'_x= N_x \xi L$ is denoted as the traffic of BS $x$.
\end{lemma}

Through following the result from Lemma $4$, both the delivery latency of F-APs user and RRHs user can be obtained. Furthermore, the average delivery latency of typical user at cache-enable FRANs is available with the assist of both latency of fronthaul links and backhaul links.

\begin{corollary}
The average delivery latency of typical user at cache-enable FRAN based on maximal RSRP association policy can be obtained as
\begin{equation}\label{file_probability}
\begin{aligned}
D& = A_F\frac{L}{\mathbb{E}\left[R_F\right]-\rho'_F}+ \left(1-p_{hit}\right)A_FD_{back}\\
&+A_R \left[ \frac{L}{\mathbb{E}\left[R_R\right]-\rho'_R}+D_{front}\right],
\end{aligned}
\end{equation}
where $\rho'_F=N_F \xi L$ and $\rho'_R=N_R \xi L$ are the traffic of F-APs and RRHs, respectively. $D_{back}=d \rho'_{back}$ and $D_{front}=d \rho'_{front}$ denote the latency of backhaul links and fronthaul links respectively, where $d$ is a constant, $\rho'_{back}=\sum\nolimits_{x \in {\Phi _F}} {\left(1-p_{hit}\right)N_F \xi L}$ is the traffic of backhaul links, and $\rho'_{front}=\sum\nolimits_{x \in {\Phi _R}} {N_R \xi L}$ is the traffic of fronthaul link.
\end{corollary}

\begin{proof}
The delivery latency of typical user at cache-enable FRAN is the weighted average among delivery latency of F-APs with the requested content, F-APs without the requested content and RRHs, which can be expressed as
\begin{equation}\label{file_probability}
\begin{aligned}
D& =p_{hit} A_F\frac{L}{\mathbb{E}\left[R_F\right]-\rho'_F}\\
&+ \left(1-p_{hit}\right)A_F\left[\frac{L}{\mathbb{E}\left[R_F\right]-\rho'_F}+D_{back}\right]\\
&+A_R \left[ \frac{L}{\mathbb{E}\left[R_R\right]-\rho'_R}+D_{front}\right],
\end{aligned}
\end{equation}
where these three items denote the weighted delivery latency of typical user at aforementioned three type of BSs, respectively. Simplifying the above equation, and thus Corollary 3 can be proved.
\end{proof}

\section{Minimal Delay Association Policy}
While the traditional maximal RSRP association policy is not suitable for the cache-enabled and latency aware F-RANs, an advanced association policy is presented, named as minimal delay association policy, to take full advantage of the edge cache and highlight the delivery latency performance while guaranteeing SE in F-RANs. Based on the minimal delay association policy, closed-form expressions of successful delivery probability, average ergodic rate, and delivery latency under cache-enabled F-RANs are proposed in this section.
\subsection{Successful Delivery Probability}
Under the minimal delay association policy, typical user is associated with the BS, which has minimal average delivery latency. Due to the introduction of caching, delivery latency is not only related to the density and power of BSs, but also the size and location of cache.

\begin{lemma}
The successful delivery probability of typical user based on minimal delay association policy can be expressed as (19) shown in the bottom at next page,
\setcounter{equation}{18}
\begin{figure*}[hb]
\hrulefill
\begin{equation}\label{file_probability}
\begin{aligned}
S\left( \left\{\eta_i\right\},\alpha \right)& = \lambda_R \int_{R^2} exp \Big( -\pi r^2 {\Big(\frac{\eta_R}{\zeta P_R}\Big)}^{2/\alpha} C \left( \alpha \right) \Big( \lambda_R \left(\upsilon P_R\right)^{2/{\alpha}}+\sum\limits_{m \in \left\{ {{F_c},{{\tilde F}_c}} \right\}} \lambda_m {P_m}^{2/{\alpha}}  \Big)\Big)exp\Big(-\frac{\eta_R r^\alpha \sigma^2}{\zeta P_R}\Big)dr\\
&+\sum\limits_{i \in \left\{ {{F_c},{{\tilde F}_c}} \right\}} \lambda_i \int_{R^2} exp \Big( -\pi r^2 {\Big(\frac{\eta_i}{P_i}\Big)}^{2/\alpha} C \left( \alpha \right) \sum\limits_{m \in \left\{ {{F_c},{{\tilde F}_c},R} \right\}} \lambda_m {P_m}^{2/{\alpha}}  \Big)exp\Big(-\frac{\eta_R r^\alpha \sigma^2}{\zeta P_R}\Big)dr,
\end{aligned}
\end{equation}
\end{figure*}
where $\lambda_{F_c}=p_{hit}\lambda_{F}$ and $\lambda_{{\tilde F}_c}=\left(1-p_{hit}\right)\lambda_{F}$ are the density of F-APs with and without requested content, respectively. $\eta_{F_c}=N_{F_c} \xi L+L/{\beta_{F_c}}$, $\eta_{{\tilde F}_c}=N_{{\tilde F}_c}\xi L+L/\left( {\beta_{{\tilde F}_c}}-D_{back} \right)$, and $\eta_{R}=N_{R}\xi L+L/\left( \beta_{R}-D_{front}\right)$ denote the SINR thresholds of F-APs with requested content, F-APs without requested content and RRHs, respectively. Here, the traffic of backhaul is replaced with $\rho'_{back}=\sum\nolimits_{x \in {\Phi _{\tilde F_c}}} {N_{\tilde F_c} \xi L}$, compared with the maximal RSRP association policy. In addition, $C \left( \alpha\right)=\left({2\pi/\alpha}\right) / sin\left({2\pi/\alpha}\right)$.
\end{lemma}

\begin{proof}
See Appendix B.
\end{proof}

Lemma 5 gives a general expression for successful delivery probability based on minimal delay association policy. Note that SINR thresholds $\eta_{i}> \zeta / \upsilon, \forall i$ are assumed, and typical user can be associated with at most one BS. Therefore, successful delivery probability can be defined as the sum of the probabilities that each BS connects to the user. This result can be simplified further and reduces to a simple closed-form expression for the interference-limited case, which is given by Corollary 4.

\begin{corollary}
When an interference-limited network is considered, the successful delivery probability of typical user based on minimal delay association policy can be simplified as
\setcounter{equation}{19}
\begin{equation}\label{file_probability}
\begin{aligned}
&S\left( \left\{\eta_i\right\},\alpha \right)=\frac{\sum\nolimits_{i \in \left\{ {{F_c},{{\tilde F}_c}} \right\}} \lambda_i \left(P_i\right)^{2/\alpha}  \left(\eta_i\right)^{-2/\alpha}}{C \left( \alpha \right) \sum\nolimits_{i \in \left\{ {{F_c},{{\tilde F}_c},R} \right\}} \lambda_i \left(P_i\right)^{2/\alpha} }\\
&+\frac{\lambda_R \left(\zeta P_R\right)^{2/\alpha}  \left(\eta_R\right)^{-2/\alpha}}{C \left( \alpha \right) \left[\lambda_R \left(\upsilon P_R\right)^{2/\alpha}+\sum\nolimits_{i \in \left\{ {{F_c},{{\tilde F}_c}} \right\}} \lambda_i \left(P_i\right)^{2/\alpha} \right]}.
 \end{aligned}
\end{equation}
\end{corollary}

\begin{corollary}
According to the successful delivery probability of typical user, association probability based on minimal delay association policy for FAPs and RRHs are given below [23].
\begin{equation}\label{file_probability}
\begin{aligned}
A_j&=\frac{\mathbb{P} \left( \bigcup\nolimits_{{x_j} \in \left\{ {\Phi_j} \right\}} {D_{x_j}< \beta_j} \right) }{Pr \left( \bigcup\nolimits_{i \in \left\{ {{F_c},{{\tilde F}_c},R} \right\},{x_i} \in \left\{ {\Phi_i} \right\}} {D_{x_i}< \beta_i} \right)}\\
&=\frac{S_{j}\left(\left\{\eta_i\right\},\alpha \right)}{S_{F_c}\left(\left\{\eta_i\right\},\alpha \right)+S_{{\tilde F}_c}\left(\left\{\eta_i\right\},\alpha \right)+S_R\left(\left\{\eta_i\right\},\alpha \right)},
\end{aligned}
\end{equation}
where $j \in \left\{ {{F_c},{{\tilde F}_c},R} \right\}$ and $S_{F_c}\left(\left\{\eta_i\right\},\alpha \right)$ is the successful delivery probability of typical user at F-APs with requested content, and it can be expressed as
\begin{equation}\label{file_probability}
\begin{aligned}
S_{F_c}\left(\left\{\eta_i\right\},\alpha \right)&=\mathbb{P} \left( \bigcup\nolimits_{{x_{F_c}} \in \left\{ {\Phi_{F_c}} \right\}} {D_{x_{F_c}}< \beta_{F_c}} \right)\\
&=\frac{ \lambda_{F_c} \left(P_{F_c}\right)^{2/\alpha}  \left(\eta_{F_c}\right)^{-2/\alpha}}{C \left( \alpha \right) \sum\nolimits_{i \in \left\{ {{F_c},{{\tilde F}_c},R} \right\}} \lambda_i \left(P_i\right)^{2/\alpha} },
\end{aligned}
\end{equation}
$S_{{\tilde F}_c}\left(\left\{\eta_i\right\},\alpha \right)$ is the successful delivery probability of typical user at F-APs without requested content, which is obtained as
\begin{equation}\label{file_probability}
\begin{aligned}
S_{{\tilde F}_c}\left(\left\{\eta_i\right\},\alpha \right)&=\mathbb{P} \left( \bigcup\nolimits_{{x_{{\tilde F}_c}} \in \left\{ {\Phi_{{\tilde F}_c}} \right\}} {D_{x_{{\tilde F}_c}}< \beta_{{\tilde F}_c}} \right)\\
&=\frac{ \lambda_{{\tilde F}_c} \left(P_{{\tilde F}_c}\right)^{2/\alpha}  \left(\eta_{{\tilde F}_c}\right)^{-2/\alpha}}{C \left( \alpha \right) \sum\nolimits_{i \in \left\{ {{F_c},{{\tilde F}_c},R} \right\}} \lambda_i \left(P_i\right)^{2/\alpha} },
\end{aligned}
\end{equation}
meanwhile, $S_R\left(\left\{\eta_i\right\},\alpha \right)$ is the successful delivery probability of typical user at RRHs, which is obtained as
\begin{equation}\label{file_probability}
\begin{aligned}
&S_R\left(\left\{\eta_i\right\},\alpha \right)=\mathbb{P} \left( \bigcup\nolimits_{{x_{R}} \in \left\{ {\Phi_{R}} \right\}} {D_{x_{R}}< \beta_{R}} \right)\\
&=\frac{\lambda_R \left(\zeta P_R\right)^{2/\alpha}  \left(\eta_R\right)^{-2/\alpha}}{C \left( \alpha \right) \left[\lambda_R \left(\upsilon P_R\right)^{2/\alpha}+\sum\nolimits_{i \in \left\{ {{F_c},{{\tilde F}_c}} \right\}} \lambda_i \left(P_i\right)^{2/\alpha} \right]}.
\end{aligned}
\end{equation}
\end{corollary}

\subsection{Average Ergodic Rate}
\begin{lemma}
The average ergodic rate of typical user based on minimal delay association policy is given by
\begin{equation}\label{file_probability}
\begin{aligned}
&R\left( \alpha \right)=\mathbb{E} \left[ log_2 \left( 1+  \mathop {\max }\limits_{x \in \bigcup\nolimits_i {{\Phi _i}}} \gamma_x \right) \right]\\
&= \int_0^\infty P \left[ \mathop {\max }\limits_{x \in \bigcup\nolimits_i {{\Phi _i}}} \gamma_x \ge {2^t-1} \right]dt\\
&=\int_0^\infty S\left( 2^t-1,\alpha \right) dt,
\end{aligned}
\end{equation}
where $S\left( {2^t-1},\alpha \right) $ is the success delivery probability of typical user based on minimal delay association policy by fixing $\eta_i = 2^t-1, \forall i$.
\end{lemma}

\subsection{Delivery Latency}
Similar to the maximal RSRP association policy, delivery latency of minimal delay association policy is also the weighted average among delivery latency of F-APs with the requested content, F-APs without the requested content and RRHs, except the weight value.

\begin{lemma}
The average delivery latency of typical user at cache-enable F-RANs based on minimal delay association policy can be obtained as
\begin{equation}\label{file_probability}
\begin{aligned}
D&= A_{F_c}\frac{L}{\mathbb{E} \left[R_F\right]-\rho'_{F_c}}+ A_{{\tilde F}_c} \left[\frac{L}{\mathbb{E} \left[R_F\right]-\rho'_{{\tilde F}_c}}+D_{back}\right]\\
&+A_R \left[ \frac{L} {\mathbb{E}\left[R_R\right]-\rho'_R}+D_{front}\right],
\end{aligned}
\end{equation}
where $A_{F_c}$, $A_{{\tilde F}_c}$ and $A_R$ denote the association probability of typical user at F-APs with the requested content, F-APs without the requested content and RRHs, respectively. $\rho'_{F_c}=N_{F_c}\xi L$, $\rho'_{{\tilde F}_c}=N_{{\tilde F}_c}\xi L$ and $\rho'_R=N_R \xi L$ are the traffic of F-APs with the requested content, F-APs without the requested content and RRHs, respectively. $D_{back}$ and $D_{front}$ denotes the latency of backhaul and fronthaul links, respectively.
\end{lemma}

\section{Numerical Results}
In this section, numerical simulations are presented to evaluate successful delivery probability, average ergodic rate and delivery latency based on both the aforementioned two association policy, while a cluster based maximal caching hit association policy with varying cluster radius is provided as the benchmark. Meanwhile, tradeoff between ergodic rate and delivery latency is investigated.

The aforementioned two-tier network is simulated in a disc plane a radius of $5 km$, and typical user is set at the origin. Simulated results are averaged over $2\times10^4$ random realizations. In detail, the distribution intensity of RRHs $\lambda_R$ is assumed to be $2\times10^{-4}/m^2$. In order to reflect the impact of cache on system performance, the intensity of F-APs is set to be $1/40$ to $1/5$ times of intensity of RRHs, i.e., $5\times10^{-6}\sim 4 \times 10^{-5}/m^2$. Service requested rate $\xi$ is assumed to be $\{5,7,9\} \times 10^{-3}$, which implies the intensity of network load.
In particular, the simulation parameters are listed as follows in Table I.

\begin{table}[t]
\caption{Simulation Parameters}
\centering
\begin{tabular}{l|l}
Parameters & Value \\
\hline
Number of contents $N$ & $50$\\
length of each content $L$ & $2$ bits\\
Caching size $C_F$& $30 \sim 70$\\
Intensity of users $\lambda_u$& $4 \times 10^{-3}/m^2$\\
Intensity of RRH nodes $\lambda_R$& $2 \times 10^{-4}/m^2$\\
Intensity of F-AP nodes $\lambda_F$& $5 \times 10^{-6} \sim 4 \times 10^{-5}/m^2$\\
Path loss exponent $\alpha$ & 4 [17] \\
Zipf exponent $\tau$ & 1\\
Transmit power of RRHs $P_R$ & 23dBm [20]\\
Transmit power of F-APs $P_F$ & 43dBm [20]\\
Service requested rate $\xi$ & $\{5,7,9\} \times 10^{-3}$\\
fronthaul/backhaul exponent $d$ & $0.5 \sim 1.5$\\
\hline
\end{tabular}
\end{table}

\begin{figure}[t]
  \centering
  \includegraphics[width=0.5\textwidth]{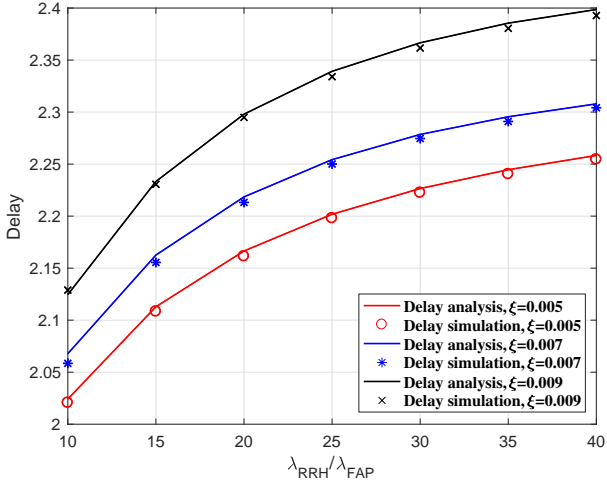}
  \vspace*{-5pt} \caption{Average delivery latency of maximal RSRP association policy with respect to RRH-FAP density ratio $\lambda_R / \lambda_F$ in the cases of different service requested rate $\xi$} \vspace*{-10pt}
  \label{1}
\end{figure}

Fig. 2 shows the average delivery latency achieved by maximal RSRP association policy with varying RRH-FAP density ratio $\lambda_R / \lambda_F$ in the cases of different service requested rate $\xi$. The analytical results closely match with the corresponding simulation results. It can be observed that the average delivery latency based on maximal RSRP association policy increases as RRH-FAP density ratio $\lambda_R / \lambda_F$ increases, this is because more users achieve services from RRHs with RRH-FAP density ratio increasing, i.e., the density of F-APs decreasing, which may bring larger delivery latency due to the latency of fronthaul links . Similarly, average delivery latency increases as service requested rate $\xi$ increases, this is because traffic density increases with service requested rate increasing, which brings larger waiting time in queuing system. Note that the service requested rate has a greater contribution to the region of higher RRH-FAP density ratio, since RRHs serve more users and are sensitive to traffic density, compared to lower RRH-FAP density ratio.

The average delivery latency of maximal RSRP association policy with varying caching hit ratio $p_{hit}$ in the cases of different service requested rate $\xi$ is shown in Fig. 3. It's obvious that average delivery latency decreases as caching hit ratio increases, since more users can achieve contents from local caches and thus mitigate the latency of fronthaul links. In addition, it is worth noting that service requested rate is not sensitive to caching hit ratio, this is because the association probability of maximal RSRP association policy depends only on the density and transmit power of both F-APs and RRHs, while caching size has no impact on the waiting time in queuing system.
\begin{figure}[t]
  \centering
  \includegraphics[width=0.5\textwidth]{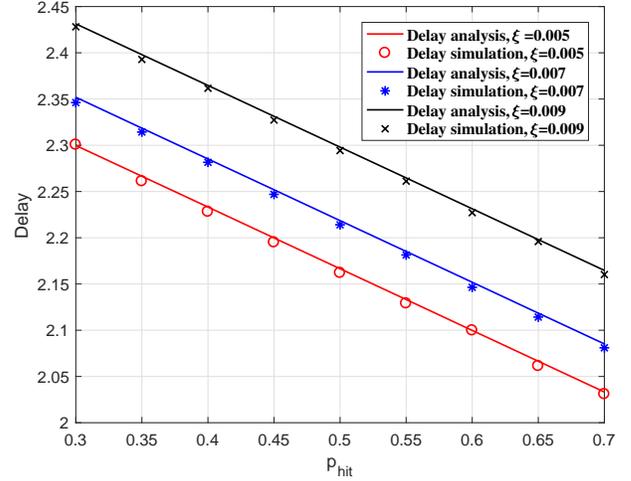}
  \vspace*{-5pt} \caption{Average delivery latency of maximal RSRP association policy with respect to caching hit ratio $p_{hit}$ in the cases of different service requested rate $\xi$} \vspace*{-10pt}
  \label{1}
\end{figure}
\begin{figure}[t]
  \centering
  \includegraphics[width=0.5\textwidth]{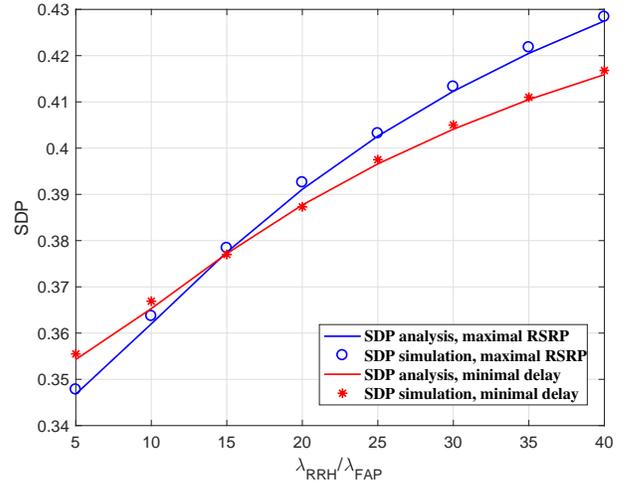}
  \vspace*{-5pt} \caption{Success delivery probability of both two user association policies with respect to RRH-FAP density ratio $\lambda_R / \lambda_F$} \vspace*{-10pt}
  \label{1}
\end{figure}

Fig. 4 and Fig. 5 show the success delivery probability and average ergodic rate achieved by both two user association policies with varying RRH-FAP density ratio $\lambda_R / \lambda_F$, respectively. It's obvious that success delivery probability increases as RRH-FAP density ratio increases in Fig. 4, since the centralized cooperative processing alleviates the interference among RRHs. Minimal delay association policy has larger success delivery probability in lower RRH-FAP density ratio region, while maximal RSRP association policy has larger success delivery probability in higher RRH-FAP density ratio region. The main reason is that users based on minimal delay association policy prefer to associate with F-APs to achieve lower delivery latency, at the cost of longer transmit distance and more intensive interference, resulting in lower success delivery probability when the number of F-APs is small. As we see from Fig. 5, the ergodic rate of minimal delay association policy is larger than that of maximal RSRP association policy in the major low RRH-FAP density ratio region, where the gap is shrunk with RRH-FAP density ratio increasing, and the ergodic rate of maximal RSRP association policy will exceed that of minimal delay association policy in extreme high RRH-FAP density ratio region, i.e., low F-APs density region, which can also be explained by the aforementioned reason. As a result, minimal delay association policy still guarantees SE demand, while pursuing lower delivery latency.

In Fig. 6, the average delivery latency based on both maximal RSRP and minimal delay association policies with varying RRH-FAP density ratio $\lambda_R / \lambda_F$ in the cases of different fronthaul link latency $D_{front}$ is investigated. As we see, minimal delay association policy is superior to maximal RSRP association policy with respect to the average delivery latency in the whole RRH-FAP density ratio region. Meanwhile, average delivery latency increases as fronthaul link latency increases. Note that minimal delay association policy is more sensitive to fronthaul link latency, since the user association of minimal delay association policy is based on delivery latency, and varying fronthaul link latency further affects average delivery latency through the corresponding changes on user association. Furthermore, fronthaul link latency has greater impact on the region of lower RRH-FAP density ratio, which has high F-APs density and more users can be accessed to F-APs due to the increasing of fronthaul link latency.
\begin{figure}[t]
  \centering
  \includegraphics[width=0.5\textwidth]{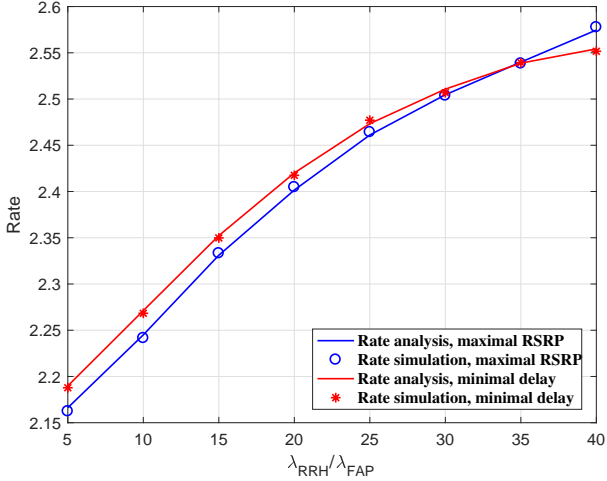}
  \vspace*{-5pt} \caption{Average ergodic rate of both two user association policies with respect to RRH-FAP density ratio $\lambda_R / \lambda_F$} \vspace*{-10pt}
  \label{1}
\end{figure}
\begin{figure}[t]
  \centering
  \includegraphics[width=0.5\textwidth]{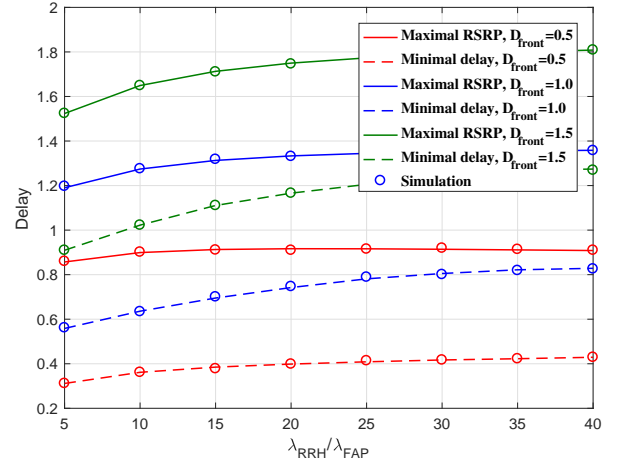}
  \vspace*{-5pt} \caption{Average delivery latency of both two user association policies with respect to RRH-FAP density ratio $\lambda_R / \lambda_F$ in the cases of different fronthaul link latency $D_{front}$} \vspace*{-10pt}
  \label{1}
\end{figure}

The average delivery latency of three user association policies with respect to RRH-FAP density ratio $\lambda_R / \lambda_F$ is shown in Fig. 7, where a cluster based maximal caching hit association policy with varying cluster radius is provided as the benchmark. With respect to this association policy, a cluster is designed as user-centered and typical user connects to the F-AP when the nearest F-AP with requested content is in the cluster. Otherwise, the user connects to RRHs.

It's obvious that the maximal RSRP association policy and minimal delay association policy are superior to the benchmarks with varying cluster radius. Meanwhile, the benchmarks are more sensitive to RRH-FAP density ratio $\lambda_R / \lambda_F$, especially for the benchmarks with large cluster radius. The average delivery latency of benchmarks increases as the cluster radiu increases, since users prefer to be associated with farther F-APs with requested content. Note that increasing $\lambda_R / \lambda_F$ increases average delivery latency based on benchmarks with large cluster radius, where associating F-APs are farther away from users and contribute to the major latency. While increasing $\lambda_R / \lambda_F$ decreases average delivery latency based on benchmarks with small cluster radius, which is dominated by the reduction of interference.

\begin{figure}[t]
  \centering
  \includegraphics[width=0.5\textwidth]{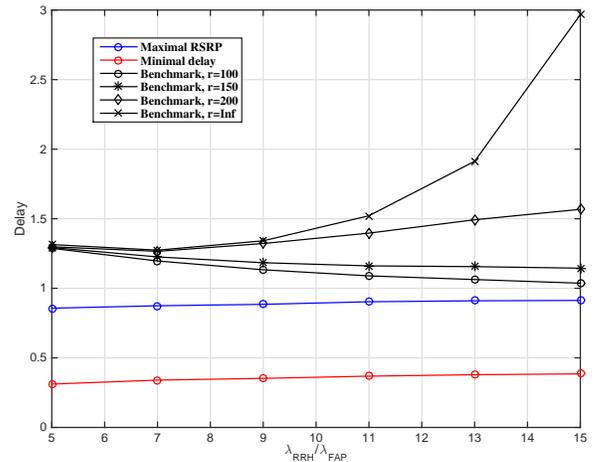}
  \vspace*{-5pt} \caption{Average delivery latency of three user association policies with respect to RRH-FAP density ratio $\lambda_R / \lambda_F$} \vspace*{-10pt}
  \label{1}
\end{figure}

\section{Conclusion}
In this paper, tradeoff between ergodic rate and delivery latency in F-RANs is explored. In order to highlight delivery latency metric, an advanced minimal delay association policy has been presented to minimize delivery latency while guaranteeing SE. By utilizing stochastic geometry, a two-tier cache-enabled F-RAN is considered according to an independent homogeneous PPP, while a M/D/1 queuing model is introduced to indicate queuing latency. As a result, closed-form expressions of successful delivery probability, average ergodic rate, and delivery latency based on both the maximal RSRP association policy and proposed minimal delay association policy are derived. The impacts of RRH-FAP density ratio, cache size, service requested rate, and fronthaul link latency on the above performance metrics are investigated. It is shown that minimal delay association policy has better delivery latency performance than maximal RSRP association policy with approximate SE. Furthermore, increasing the cache size of F-APs can more significantly reduce average delivery latency, compared with increasing the density of F-APs, since the latter may decrease average ergodic rate. Meanwhile, it is worth mentioning that, like the HD video services, other kinds of services with their specific caching strategies are also applicable in this paper.

\begin{appendices}
\section{}
We first derive the success probability conditioned on the typical user associated with F-APs, i.e., $k{r_{x_0}}^ {-\alpha} \le {r_{y_0}}^ {-\alpha}$, and the conditional success probability of typical user at F-APs is given by (27) shown in the bottom at next page,
\setcounter{equation}{26}
\begin{figure*}[hb]
\hrulefill
\begin{equation}\label{file_probability}
\begin{aligned}
S_F\left( \delta,\alpha \right)&=\mathbb{P} \Big( \frac{P_F{\left| h_{x_0} \right|}^2{r_{x_0}}^ {-\alpha}} {I_{F}+I_R+{\sigma}^2} \ge \delta |k{r_{x_0}}^ {-\alpha} \le {r_{y_0}}^ {-\alpha} \Big)=\int_0^\infty \mathbb{P} \Big({\left| h_{x_0} \right|}^2 \ge \frac{\delta {r_{x_0}}^ {\alpha}}{P_F} \left( I_{F}+I_R \right) \Big) f_F \left( r|k{r_{x_0}}^ {-\alpha} \le {r_{y_0}}^ {-\alpha}\right) dr\\
&\mathop  = \limits^{\left( a \right)}\int_0^\infty { \mathbb{E}\Big[ exp \Big( -\frac{\delta r^{\alpha}I_{F}}{P_F} \Big) \Big]}\mathbb{E} \Big[ exp \Big(-\frac{\delta r^{\alpha}I_{R}}{P_F} \Big) \Big] f_F \left( r|k{r_{x_0}}^ {-\alpha} \le {r_{y_0}}^ {-\alpha} \right)dr\\
&\mathop  = \limits^{\left( b \right)}\int_0^\infty exp\left( -\pi r^2 \lambda_F \rho\left( \delta, \alpha \right) \right) exp\Big( -\pi \left( kr \right)^2 \lambda_R \rho\Big(  \frac{\delta P_R}{k^{\alpha}P_F}, \alpha \Big) \Big)2\pi r\left( \lambda_F +k^2 \lambda_R\right)exp\left(-\pi r^2\left(\lambda_F +k^2 \lambda_R\right)\right)dr,
\end{aligned}
\end{equation}
\end{figure*}
\setcounter{equation}{28}
\begin{figure*}[hb]
\begin{equation}\label{file_probability}
\begin{aligned}
S_R\left( \delta,\alpha \right)&=\mathbb{P} \Big(\frac{ P_R \zeta {\left| h_{y_0} \right|}^2{r_{y_0}}^ {-\alpha}} {I_{F}+\upsilon I_{R}+{\sigma}^2} \ge \delta |k{r_{x_0}}^ {-\alpha} > {r_{y_0}}^ {-\alpha} \Big)=\int_0^\infty \mathbb{P} \Big({\left| h_{y_0} \right|}^2 \ge \frac{\delta {r_{y_0}}^ {\alpha}}{\zeta P_R } \left( I_{F}+\upsilon I_R \right) \Big) f_R \left( r|k{r_{x_0}}^ {-\alpha} > {r_{y_0}}^ {-\alpha}\right) dr\\
&\mathop  = \limits^{\left( a \right)}\int_0^\infty {\mathbb{E}\Big[ exp \Big( -\frac{\delta r^{\alpha}I_{F}}{\zeta P_R} \Big) \Big]}\mathbb{E} \Big[ exp \Big(-\frac{\delta r^{\alpha} \upsilon I_{R}}{\zeta P_R} \Big) \Big] f_R \left( r|k{r_{x_0}}^ {-\alpha} > {r_{y_0}}^ {-\alpha} \right)dr\\
&\mathop  = \limits^{\left( b \right)}\int_0^\infty exp\Big( -\pi \frac {r^2}{k^2} \lambda_F \rho \Big( \frac{\delta k^{\alpha}P_F}{\zeta P_R} \Big) \Big)exp\Big( -\pi r^2 \lambda_R \rho\Big(  \frac{\upsilon \delta }{\zeta}, \alpha \Big) \Big)2\pi r\Big( \frac{\lambda_F} {k^2} +\lambda_R\Big) exp\Big( -\pi r^2\Big( \frac{ \lambda_F}{k^2} +\lambda_R\Big)\Big)dr,
\end{aligned}
\end{equation}
\end{figure*}
where $\left( a \right)$ follows from the Laplace transform of $ {\left| h_{x_0} \right|}^2 \sim Exp \left( 1 \right)$ and the independence between intra-tier interference $I_F$ and inter-tier interference $I_R$; (b) follows from letting $s=\frac{\delta r^{\alpha}}{P_F}$ in the Laplace transforms of $I_F$ and $I_R$, the probability density function (PDF) of $r_{x_0}$ conditioned on the event $k{r_{x_0}}^ {-\alpha} \le {r_{y_0}}^ {-\alpha}$ is $f_F \left( r|k{r_{x_0}}^ {-\alpha} \le {r_{y_0}}^ {-\alpha} \right)=2\pi r\left( \lambda_F +k^2 \lambda_R\right)exp\left(-\pi r^2\left(\lambda_F +k^2 \lambda_R\right)\right)$.

Note that interference limited channel is considered since the interference is much larger than the noise, i.e., $\sigma^2 \to 0$, and $\rho \left( \delta,\alpha \right)=\int_{{\delta ^{ - 2 / \alpha}}}^\infty  {\frac{{{\delta ^{2/\alpha }}}}{{1 + {u^{\alpha /2}}}}} du$.

Therefore, the conditional success probability of typical user at F-APs can be expressed as
\setcounter{equation}{27}
\begin{equation}\label{file_probability}
S_F\left( \delta,\alpha \right)=\frac{\lambda_F+k^2 \lambda_R} {\lambda_F \left( 1+ \rho \left( \delta,\alpha \right) \right)+ k^2 \lambda_R \left( 1+ \rho \left( \frac{\delta P_R}{k^{\alpha}P_F},\alpha \right) \right) }.
\end{equation}

Otherwise, when typical user is associated with RRHs, i.e., $k{r_{x_0}}^ {-\alpha} > {r_{y_0}}^ {-\alpha}$, the conditional success probability of typical user at RRHs is given by (29) shown in the bottom at next page, where $\left( a \right)$ follows from the Laplace transform of $ {\left| h_{y_0} \right|}^2 \sim Exp \left( 1 \right)$ and the independence between intra-tier interference $I_R$ and inter-tier interference $I_F$; (b) follows from letting $s=\frac{\delta r^{\alpha}}{\zeta P_R}$ in the Laplace transforms of $I_F$ and $I_R$, the PDF of $r_{y_0}$ conditioned on the event $k{r_{x_0}}^ {-\alpha} > {r_{y_0}}^ {-\alpha}$ is $f_R \left( r|k{r_{x_0}}^ {-\alpha} > {r_{y_0}}^ {-\alpha} \right)=2\pi r\left( \frac{\lambda_F} {k^2} +\lambda_R\right) exp\left( -\pi r^2\left( \frac{ \lambda_F}{k^2} +\lambda_R\right)\right)$, and interference limited channel is considered, i.e., $\sigma^2 \to 0$.

Therefore, the conditional success probability of typical user at RRHs can be expressed as
\setcounter{equation}{29}
\begin{equation}\label{file_probability}
S_R\left( \delta,\alpha \right)=\frac{\lambda_F+k^2 \lambda_R} {\lambda_F \left( 1+ \rho \left( \delta_1,\alpha \right) \right)+ k^2 \lambda_R \left( 1+ \rho\left(  \delta_2, \alpha \right)  \right) }.
\end{equation}
where $\delta_1=\frac{\delta k^{\alpha}P_F}{\zeta P_R}$, and $\delta_2=\frac{\upsilon \delta }{\zeta}$.

As a result, utilizing (28) and (30), the successful delivery probability of typical user can be obtained.

\section{}
The successful delivery probability of typical user based on minimal delay association policy can be derived as follows
\begin{equation}\label{file_probability}
\begin{aligned}
 &S\left( \left\{\eta_i\right\},\alpha \right)=\mathbb{P} \Big( \bigcup\limits_{i \in \left\{ {{F_c},{{\tilde F}_c},R} \right\},{x_i} \in \left\{ {\Phi_i} \right\}} {D_{x_i}< \beta_i} \Big)\\
 &=\mathbb{E} \Big[ \pmb{1} \Big( \bigcup\limits_{i \in \left\{ {{F_c},{{\tilde F}_c},R} \right\},{x_i} \in \left\{ {\Phi_i} \right\}} {D_{x_i}< \beta_i} \Big) \Big]\\
&=\mathbb{E} \Big[ \pmb{1} \Big( \bigcup\limits_{i \in \left\{ {{F_c},{{\tilde F}_c},R} \right\},{x_i} \in \left\{ {\Phi_i} \right\}} {SINR_{x_i}> \eta_i} \Big) \Big],
 \end{aligned}
\end{equation}

It should be noted that the SINR threshold of typical user at F-APs is assumed to be $\eta_F>1$ and the SINR threshold of typical user at RRHs is assumed to be $\eta_R> \zeta / \upsilon$, where $\zeta$ and $\upsilon$ are the statistic of quantized coefficient and interference power coefficient, respectively. As a result, each user can be associated with at most one BS, and successful delivery probability can be translated into the sum of the successful probabilities that each BS connects to the typical user.

Therefore, the successful delivery probability of typical user based on minimal delay association policy can be simplified as
\begin{equation}\label{file_probability}
\begin{aligned}
&S\left( \left\{\eta_i\right\},\alpha \right)=\mathbb{E} \Big[ \sum\limits_{{x_R} \in \left\{ {\Phi_R} \right\}} \left[ \pmb{1} \left( {SINR_{x_R}> \eta_R} \right) \right]\Big]\\
&+\sum\limits_{i \in \left\{ {{F_c},{{\tilde F}_c}} \right\}} \mathbb{E}\Big[ \sum\limits_{{x_i} \in \left\{ {\Phi_i} \right\}} \left[ \pmb{1} \left( {SINR_{x_i}> \eta_i} \right) \right]\Big],
 \end{aligned}
\end{equation}
where the first part denotes the successful delivery probability of typical user at RRHs, and the second part denotes the successful delivery probability of typical user at F-APs with or without the requested content.

Specifically, the successful delivery probability of typical user at F-APs can be derived as
\begin{equation}\label{file_probability}
\begin{aligned}
&\sum\limits_{i \in \left\{ {{F_c},{{\tilde F}_c}} \right\}} \mathbb{E}\Big[ \sum\limits_{{x_i} \in \left\{ {\Phi_i} \right\}} \left[ \pmb{1} \left( {SINR_{x_i}> \eta_i} \right) \right]\Big]\\
&=\sum\limits_{i \in \left\{ {{F_c},{{\tilde F}_c}} \right\}} \lambda_i \int_{R^2} {{\cal L}_{I_{x_i}} \left( \eta_i r^{\alpha} {P_i}^{-1} \right)}exp\Big(-\frac{\eta_R r^\alpha \sigma^2}{ P_R}\Big)dr\\
&=\sum\limits_{i \in \left\{ {{F_c},{{\tilde F}_c}} \right\}} \lambda_i \int_{R^2} exp  \Big(-\pi r^2 {\left(\frac{\eta_i}{P_i}\right)}^{2/\alpha}C \left( \alpha \right)\\
&\times \sum\limits_{m \in \left\{ {{F_c},{{\tilde F}_c},R} \right\}} \lambda_m {P_m}^{2/{\alpha}}  \Big)exp\Big(-\frac{\eta_R r^\alpha \sigma^2}{P_R}\Big)dr,
\end{aligned}
\end{equation}
where $\eta_{{\tilde F}_c}=N_{{\tilde F}_c}\xi L+L/\left( {\beta_{{\tilde F}_c}}-D_{back} \right)$. Here, $D_{back}=d_{back}\rho'_{back}$ denotes the latency of backhaul links, $d_{back}=d$ is a constant, $\rho'_{back}=\sum\nolimits_{x \in {\Phi _{\tilde F_c}}} {N_{\tilde F_c} \xi L}$ is the traffic of backhaul links, $C \left( \alpha\right) =\left({2\pi/\alpha}\right) / sin\left({2\pi/\alpha}\right)$.

Similarly, the successful delivery probability of typical user at RRHs can be derived as
\begin{equation}\label{file_probability}
\begin{aligned}
&\mathbb{E}\Big[ \sum\limits_{{x_R} \in \left\{ {\Phi_R} \right\}} \left[ \pmb{1} \left( {SINR_{x_R}> \eta_R} \right) \right]\Big]\\
&=\lambda_R \int_{R^2} {{\cal L}_{I_{x_R}} \left( \eta_R r^{\alpha} {\zeta P_R}^{-1} \right)}exp\Big(-\frac{\eta_R r^\alpha \sigma^2}{\zeta P_R}\Big)dr\\
&=\lambda_R \int_{R^2} exp \Big( -\pi r^2 {\left(\frac{\eta_R}{\zeta P_R}\right)}^{2/\alpha}C \left( \alpha \right)\Big( \lambda_R \left(\upsilon  P_R\right)^{2/{\alpha}}\\
& + \sum\limits_{m \in \left\{ {{F_c},{{\tilde F}_c}} \right\}} \lambda_m {P_m}^{2/{\alpha}}  \Big)\Big) exp\Big(-\frac{\eta_R r^\alpha \sigma^2}{\zeta P_R}\Big)dr,
\end{aligned}
\end{equation}
where $\eta_{R}=N_{R}\xi L+L/\left( \beta_{R}-D_{front}\right)$, and $D_{front}=d_{front}\rho'_{front}$ denotes the latency of fronthaul links, $d_{front}=d$ is a constant, $\rho'_{front}=\sum\nolimits_{x \in {\Phi _R}} {N_R \xi L}$ is the traffic of fronthaul links.
\end{appendices}

\end{document}